\documentclass[letterpaper]{article} % DO NOT CHANGE THIS
\usepackage{aaai24}  % DO NOT CHANGE THIS
\usepackage{amsmath}
\usepackage{amsfonts}
\usepackage{subcaption}

\usepackage{times}  % DO NOT CHANGE THIS
\usepackage{helvet}  % DO NOT CHANGE THIS
\usepackage{courier}  % DO NOT CHANGE THIS
\usepackage[hyphens]{url}  % DO NOT CHANGE THIS
\usepackage{graphicx} % DO NOT CHANGE THIS
\usepackage{natbib}  % DO NOT CHANGE THIS AND DO NOT ADD ANY OPTIONS TO IT
\usepackage{caption} % DO NOT CHANGE THIS AND DO NOT ADD ANY OPTIONS TO IT
\usepackage{algorithm}
\usepackage{algorithmic}

\usepackage{newfloat}
\usepackage{listings}
\DeclareCaptionStyle{ruled}{labelfont=normalfont,labelsep=colon,strut=off} % DO NOT CHANGE THIS
\floatstyle{ruled}
\newfloat{listing}{tb}{lst}{}
\floatname{listing}{Listing}
\title{Graph Bayesian Optimization for Multiplex Influence Maximization}
\author {
    Zirui Yuan\textsuperscript{\rm 1},
    Minglai Shao\textsuperscript{\rm 1}\thanks{Corresponding author},
    Zhiqian Chen\textsuperscript{\rm 2}
}
\affiliations {
    \textsuperscript{\rm 1}School of New Media and Communication, Tianjin University, China\\
    \textsuperscript{\rm 2}Department of Computer Science and Engineering, Mississippi State University,  USA \\
    yzr@tju.edu.cn, shaoml@tju.edu.cn, zchen@cse.msstate.edu
}

\usepackage{bibentry}
\usepackage[english]{babel}
\usepackage{amsmath}
\usepackage{amsthm}
\usepackage{amssymb}
\theoremstyle{definition}
\newtheorem{definition}{Definition}
\theoremstyle{plain} % default style
\newtheorem{theorem}{Theorem}

\begin{document}

\maketitle

\begin{abstract}
Influence maximization (IM) is the problem of identifying a limited number of initial influential users within a social network to maximize the number of influenced users. However, previous research has mostly focused on individual information propagation, neglecting the simultaneous and interactive dissemination of multiple information items. In reality, when users encounter a piece of information, such as a smartphone product, they often associate it with related products in their minds, such as earphones or computers from the same brand. Additionally, information platforms frequently recommend related content to users, amplifying this cascading effect and leading to multiplex influence diffusion.

This paper first formulates the Multiplex Influence Maximization (Multi-IM) problem using multiplex diffusion models with an information association mechanism. In this problem, the seed set is a combination of influential users and information. To effectively manage the combinatorial complexity, we propose Graph Bayesian Optimization for Multi-IM (GBIM). The multiplex diffusion process is thoroughly investigated using a highly effective global kernelized attention message-passing module. This module, in conjunction with Bayesian linear regression (BLR), produces a scalable surrogate model. A data acquisition module incorporating the exploration-exploitation trade-off is developed to optimize the seed set further.
Extensive experiments on synthetic and real-world datasets have proven our proposed framework effective. The code is available at \url{https://github.com/zirui-yuan/GBIM}.
\end{abstract}

\section{Introduction}
The rapid growth of online social networks has sparked substantial interest among researchers in understanding the dynamics of information dissemination within these networks. This interest has led to the study of influence maximization (IM), an optimization problem that aims to maximize the influence spread within a specific diffusion model by selecting a limited number of seeds \cite{kempe2003maximizing}. IM has garnered considerable attention from both industry and academia due to its relevance in various real-world applications, including viral marketing, epidemic control, and rumor blocking \cite{li2018influence, zhang2022blocking,li2022survey}. 
\begin{figure}[t]
\centering
\includegraphics[width=\linewidth]{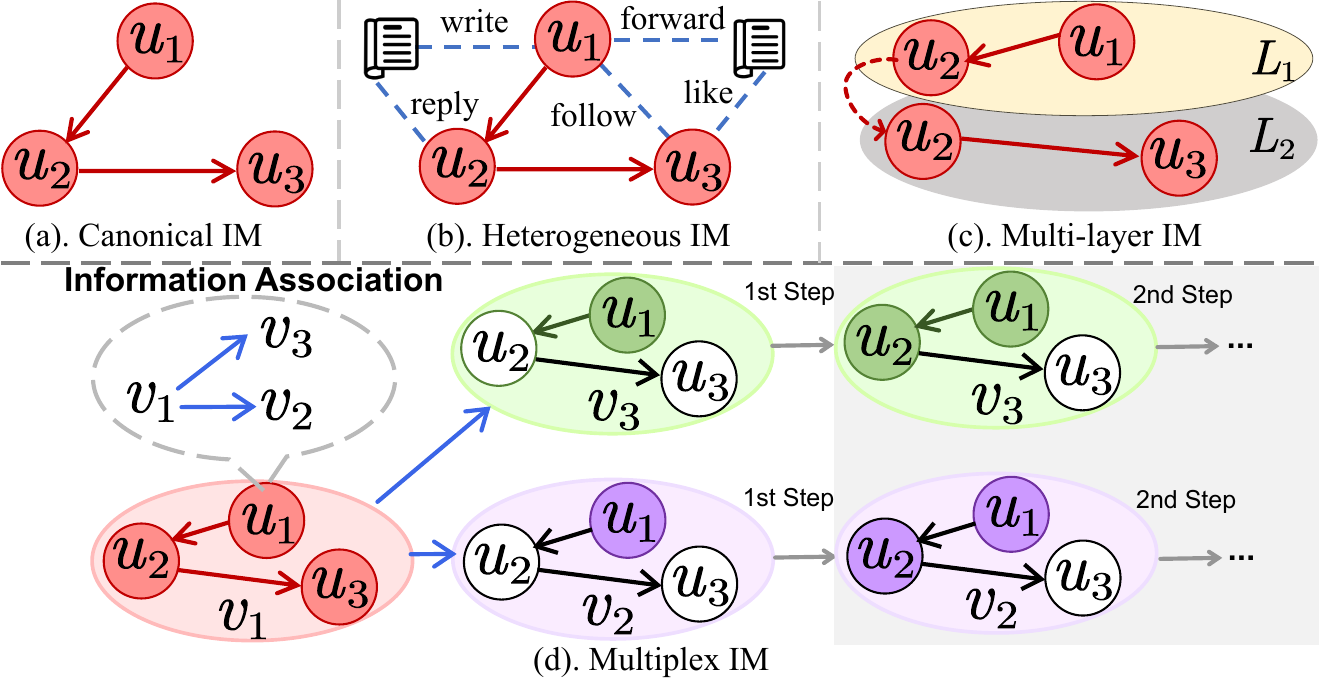}
	\caption{The comparison of related IM problems, and $u_1$ is the initial seed user. (a) Canonical IM in the homogeneous network. (b) Heterogeneous IM on networks with multiple node and edge types. (c) Multi-layer IM with users in two social network platforms. (d) Multiplex IM with three information items $\{v_1, v_2, v_3\}$, and $(u_1, v_1)$ is the seed.}
\label{fig:im}
\end{figure}

Most IM literature focuses on the canonical setting, which only considers an individual information item spread over a homogeneous social network based on an information diffusion model. Researchers have made great progress in devising various efficient and effective methods to address the canonical IM problem, utilizing techniques such as simulation, proxy, and sketch \cite{goyal2011celf++, borgs2014maximizing, tang2015influence}. 
In recent years, IM research has witnessed a growing interest in exploring more realistic scenarios beyond the canonical setting, including \textit{Heterogeneous} IM and \textit{Multi-Layer} IM. 
Heterogeneous IM (Figure \ref{fig:im}b) addresses the IM problem in heterogeneous networks with diverse node and edge types. Researchers commonly employ techniques involving network schema and meta-paths to construct propagation pathways. Building on this, both traditional and learning-based methods can be designed to optimize the seed set \cite{zhan2015influence, deng2019influence, li2022mahe}.
Multi-Layer IM (Figure \ref{fig:im}c) discusses the cross-layer nature of information dissemination across different social network platforms. These studies consider the sharing of nodes between networks as a mechanism enabling information to propagate across layers. Some approaches for tackling this problem involve integrating networks into one unified large network and applying existing IM algorithms \cite{kuhnle2018multiplex, keikha2020influence, katukuri2022cim}. 

While these efforts have succeeded, they still concentrate on scenarios where individual information spreads in isolation. However, real-world online social networks typically involve the simultaneous dissemination of multiple information items, with heterogeneous diffusion patterns over these items. Additionally, interconnected information items often result in multiplex influence, further complicating the dynamics of information spread. This study introduces the Multiplex Influence Maximization (Multi-IM) problem, which aims to maximize the multiplex influence within the constraints of a fixed seed budget. As depicted in Figure \ref{fig:im}d, three interconnected information items represented by $\{v_1, v_2, v_3\}$ (e.g., three complementary products) propagate concurrently in a social network. This can be modeled as a multiplex network consisting of three propagation layers, each representing the dissemination of one information item, denoted by the colors red, purple, and green, respectively.
When information item $v_1$ influences user $u_1$, it creates multiplex influence due to the associations among the three information items. Consequently, user $u_1$ also adopts the information from $v_2$ and $v_3$, activating $u_1$ on the propagation layers of $v_2$ and $v_3$ in this multiplex network. Moreover, the preference biases of users towards items lead to different propagation results.

The new mode of propagation presents the following difficulties for seed set selection. 
\textbf{First}, the presence of multiple information items increases combinatorial complexity. To maximize multiplex influence, we must not only select users to disseminate information but also determine the most pertinent information items for those users.
\textbf{Second}, this novel modeling of propagation results in a heterogeneous, multidimensional structure that introduces a new type of geometric complexity.
\textbf{Furthermore}, compared to Canonical IM, this propagation modeling generates larger dimensions and data volume, resulting in catastrophic data processing issues. Numerous layers representing item-specific propagation environments would considerably increase computational costs in Multi-IM, making it challenging for conventional techniques to complete tasks in a timely manner.

We propose GBIM, a Graph Bayesian Optimization framework for Multi-IM, to tackle these challenges. By effectively navigating the complex search space, GBIM identifies the most influential users and the optimal information items to maximize multiplex influence. 
To address the heterogeneous and multi-layered geometric structure generated by the new propagation modeling, a highly effective global kernelized attention message-passing module based on Positive Random Features (PRF) \cite{choromanski2020rethinking} is employed. This module acquires a comprehensive understanding of the multiplex diffusion process, yielding a reliable influence estimation. GBIM combines this module with Bayesian linear regression to create a scalable surrogate model that facilitates efficient data processing. In addition, the data acquisition module incorporates the exploration-exploitation trade-off, enabling effective exploration of the vast search space for optimizing the seed set with the highest observed performance.
Our contributions are as follows:
% \textbf{ (1)} We first formulate the Multi-IM problem, which models the scenario that multiple information items as a multiplex network. 
% \textbf{(2)} We devise a highly effective global kernelized attention message-passing module to learn the complex multiplex diffusion process.
% \textbf{(3)} We present GBIM to optimize the seed set with the highest observed performance efficiently.
% \textbf{(4)} We conduct extensive experiments to demonstrate the performance of the proposed method.
\begin{itemize}
    \item First formulate the Multi-IM problem, which models the scenario that multiple information items propagate and interact in a multiplex network.
    \item Devise a highly effective global kernelized attention message-passing module to learn the complex multiplex diffusion process.
    \item Present GBIM to optimize the seed set with the highest observed performance efficiently.
    \item Conduct extensive experiments to demonstrate the performance of the proposed method.
\end{itemize}

\section{Related Work}
\subsubsection{Influence Maximization.}
The influence maximization (IM) problem was introduced in seminal work by \cite{kempe2003maximizing}, relying on diffusion models like Linear Threshold and Independent Cascade \cite{li2018influence,li2022survey}. Much research has developed traditional simulation, proxy, and sketch-based methods \cite{goyal2011celf++, borgs2014maximizing, tang2015influence}, as well as recent learning-based techniques \cite{ling2023deep}. However, these canonical IM studies focus on single information spread in homogeneous networks.
Recent interest has emerged in more realistic settings like Heterogeneous IM and Multi-Layer IM. Heterogeneous IM handles diverse nodes and edges using techniques involving meta-paths \cite{zhan2015influence, deng2019influence, li2022mahe}. Multi-Layer IM examines cross-platform information spread via node sharing \cite{kuhnle2018multiplex, keikha2020influence, katukuri2022cim}. While incorporating complex networks, these works still consider individual information.
Multi-information settings are first discussed in Competitive IM, which handles a purely competitive scenario \cite{bharathi2007competitive}. The complement relation between information items is only investigated in \cite{lu2015competition} but is limited to the IC model. Recently, some studies have examined multi-information scenarios but did not account for interconnections between information items \cite{ni2020information, wu2021parallel, fang2022greedy}.
In summary, modeling complex interactions among multiple information items remains a key limitation. Our work introduces a multiplex network perspective to model multi-information dissemination.

\subsubsection{Bayesian Optimization.}
Bayesian optimization (BO) \cite{mockus1998application} is a method to optimize complex black-box functions relying on Bayesian statistical models paired with optimization algorithms. It requires a surrogate model representing the target function, typically based on Gaussian processes (GPs) \cite{brochu2010bayesian, snoek2012practical}. However, GPs scale cubically with increasing observations, challenging massively parallel optimization requiring many evaluations. 
Neural networks can serve as a practical surrogate model, leveraging flexible representation power to model complex functions and enhance scalability \cite{snoek2015scalable,springenberg2016bayesian, perrone2018scalable,ma2019deep}.  
In this work, we develop a highly effective global kernelized attention message-passing module to learn the multiplex diffusion process and serve as non-linear basis functions for Bayesian linear regression to yield a scalable surrogate model, leveraging both the nonlinear fitting capabilities of neural networks and the effective statistical properties of Bayesian.

\section{Problem and Model Definition}
In this work, we examine a context where $m$ distinct pieces of information are concurrently disseminated across a weighted, directed social network, represented as $\mathcal{G}=(\mathcal{U}, \mathcal{E}_\mathcal{U})$. 
Here, $\mathcal{U} = \{u_1,u_2,...,u_n\}$ signifies the user set, while $\mathcal{E}_\mathcal{U}$ denotes the edge set, with each $e_{i,j} \in \mathcal{E}_\mathcal{U}$ carrying a specific weight $w_{i,j}$. We introduce the information association network as $\mathcal{I} = (\mathcal{V}, \mathcal{E}_\mathcal{V})$, where $\mathcal{V} = \{v_1,v_2,...,v_m\}$ represents the set of information items and $\mathcal{E}_v$ encompasses the edges linking these items. Furthermore, we consider a preference matrix, $\mathbf{P} \in \mathbb{R}^{n \times m}$, in which the element $p_{i,j}$ reflects the inclination of user $u_i$ towards item $v_j$. The multiplex diffusion model $\mathcal{M}$ is expressed as $y = \mathcal{M}(\boldsymbol{x};\mathcal{G},\mathcal{I},\mathbf{P})$. In this representation, the input $\boldsymbol{x} = \{...,(u,v),...\}$ indicates a seed set containing multiple user-item pairs, while each pair $(u,v) \in \boldsymbol{x}$ implies that user $u$ is initially influenced by the information item $v$ during the execution of $\mathcal{M}$. The output $y \in \mathbb{N}^+$ signifies the multiplex influence exerted by all items. 
Grounded on the formalization above, the Multi-IM problem is defined as follows:

\begin{definition}[Multiplex Influence Maximization]
\textit{The objective of the Multi-IM problem is to strategically choose a maximum of $k$ user-item pair from $\mathcal{U}$ and $\mathcal{V}$, each user and item can only be selected once, so as to maximize the overall multiplex influence.}
\begin{equation}
\boldsymbol{x}^*=\arg\max_{|\boldsymbol{x}| \le k} \mathcal{M}(\boldsymbol{x};\mathcal{G},\mathcal{I},\mathbf{P}),
\end{equation}
\begin{equation*}
\text{s.t.}\quad u_i \neq u_j,  v_i \neq v_j, \quad\forall (u_i,v_i),(u_j,v_j) \in \boldsymbol{x},
\end{equation*}
\textit{where $\boldsymbol{x}^*$ represents the optimal seed node set capable of generating the maximum multiplex influence in $\mathcal{M}$.}
\end{definition}
In reality, the propagation dynamics of each information item are different since users are inclined to spread information aligning with their interests while overlooking content that does not appeal to their preferences. Besides, when users are influenced, they may associate other interconnected information in their minds, resulting in multiplex influence.
To model this realistic scenario, we introduce the multiplex influence diffusion model $\mathcal{M}$ along with the association mechanism in a multiplex network perspective.

% \begin{definition}[Multiplex Network]
% \textit{A multiplex network is a graph with multiple layers, where each layer contains a network on the same set of nodes.}
% \end{definition}
A multiplex network is a graph with multiple layers, where each layer contains a network on the same set of nodes.
The multiplex network in $\mathcal{M}(\boldsymbol{x};\mathcal{G},\mathcal{I},\mathbf{P})$ contains $m$ layers, where each layer models the propagation environment for one information item and shares the same structure as $\mathcal{G}$. The commonly used information diffusion models like Linear Threshold (LT) and Independent Cascade (IC) can simulate propagation on each layer but need adaption to incorporate inherent heterogeneity stemming from preference matrix $\mathbf{P}$.
Specifically, within the propagation layer of information item $v_k$, the threshold for user $u_i$ is adjusted to $1-p_{i,k}$ when using the LT model. For the IC model, the influence probability of edge $e_{i,j}$ is modified to $w_{i,j} \cdot p_{j,k}$.
We propose an association mechanism to model the inter-layer propagation across information items:
\begin{definition}[Association Mechanism]
\textit{The Association Mechanism in the multiplex influence diffusion model refers to the process by which activated users spread across the interconnected information item propagation layers. }
\end{definition}
Specifically, when user $u_i$ gets activated by information $v_j$, they have a probability $\beta*p_{i,k}$ of associating to other adjacent information $v_k \in \mathcal{N}_I(v_j)$ and self-activating in $v_k$'s propagation layer. Here, $\mathcal{N}_I(v_j)$ denotes the neighbor information items of $v_j$ in $\mathcal{I}$, and $\beta$ is a base ratio scaling the inter-layer association strength. This process resembles diffused thinking along $\mathcal{I}$, terminating when user $u_i$ makes no more associations.

Building on the multiplex network structure and association mechanism, we formally define the multiplex influence model as follows:

\begin{definition}[Multiplex Diffusion Model]
\textit{The multiplex diffusion model $\mathcal{M}(\boldsymbol{x};\mathcal{G},\mathcal{I},\mathbf{P})$ is a model characterizing the dynamical diffusion process of multiple information items in a multiplex network. Each layer has a heterogeneous diffusion pattern depending on the preference matrix $\mathbf{P}$, and the association mechanism allows influenced users to propagate across layers based on $\mathcal{I}$. The output is the expected activated number at the multiplex network.}
\end{definition}
Let $\sigma_{i}(\boldsymbol{x})$ denote the number of the users influenced by item $v_i$, we define the multiplex influence output by $\mathcal{M}$ as:
\begin{equation}
 \mathcal{M}(\boldsymbol{x};\mathcal{G},\mathcal{I},\mathbf{P}) = \sum_{i=1}^m{\sigma_{i}(\boldsymbol{x})}.  
\label{eq:mtinf}
\end{equation}
The multiplex diffusion model $\mathcal{M}(\boldsymbol{x})$ operates in discrete time steps. At step 0, the starting state is configured based on the input seed set $\boldsymbol{x}$. In each following step, the association mechanism is executed for every newly activated user. Heterogeneous information diffusion then takes place on propagation layers with active users. The model repeats this process until no further activations occur and outputs the final multiplex influence.

\begin{figure*}[t]
\centering
\includegraphics[width=\textwidth]{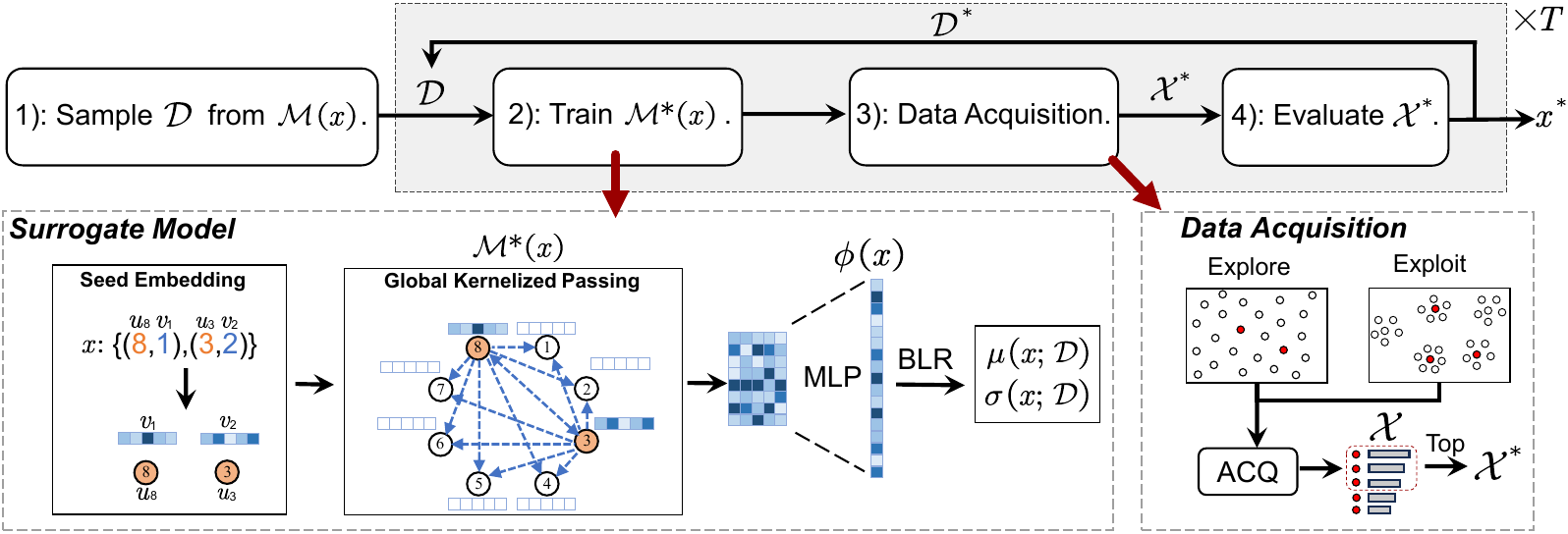}
\caption{The overview of the proposed GBIM framework. This framework includes two modules: surrogate model and data acquisition. Initially, random seeds are evaluated by the true model $\mathcal{M}(\boldsymbol{x})$ to form dataset $\mathcal{D}$. We then iteratively: (1) train surrogate model $\mathcal{M}^*(\boldsymbol{x})$ on $\mathcal{D}$; (2) sampling and evaluate the candidates $\mathcal{X}$ via $\mathcal{M}^*(\boldsymbol{x})$, selecting top K seed sets to be $\mathcal{X}^*$; (3) assess $\mathcal{X}^*$ via $\mathcal{M}(\boldsymbol{x})$ to expand $\mathcal{D}$. Finally, the optimal $\boldsymbol{x}^*$ from $\mathcal{D}$ with maximal influence is selected.}
\label{fig:gbim_arch}
\end{figure*}

\section{GBIM: Graph Bayesian Optimization Framework for Multi-IM}
\subsection{Framework Overview}
% As illustrated by Figure 1, our proposed GBIM framework consists of three main components: GNN-based surrogate model, acquisition function, and sampling strategy. 
% The workflow of our proposed GBIM framework is described in Figure 1. The GBIM framework contains two main components: the surrogate function and data acquisition. Firstly, we randomly generate seeds and evaluate with the objective multiplex influence model $\mathcal{M}(x,\mathcal{G};\theta)$ to construct the initial training data $\mathcal{D} = \{(x_1,y_1),...,(x_N,y_N)\}$. Then, we repeatedly execute the following steps: (1) train a surrogate information diffusion model $\mathcal{M}^*(x|\mathcal{D})$ to fit the current data $\mathcal{D}$; (2) use the output from the last layer of the $\mathcal{M}^*(x|\mathcal{D})$ as basis function and add a Bayesian linear regressor (BLR) to select a new set of data $\mathcal{X}^*$ with top K high acquisition value; (3) evaluate $\mathcal{X}^*$ at $\mathcal{M}(x,\mathcal{G};\theta)$ to obtain a new dataset and augmented the training data. After the optimization loop terminates, we select the $x_best$ in the current observed data $\mathcal{D}$ with the max spread num as the final seed set. 

The process of our suggested GBIM framework is illustrated in Figure \ref{fig:gbim_arch}. This framework encompasses two primary modules: the surrogate model and data acquisition. Initially, seeds are generated at random and assessed using the objective multiplex diffusion model $\mathcal{M}(\boldsymbol{x})$, forming the preliminary training dataset $\mathcal{D} = \{(\boldsymbol{x}_1,y_1),...,(\boldsymbol{x}_N,y_N)\}$. Subsequently, we iteratively: (1) train a surrogate diffusion model, denoted as $\mathcal{M}^*(\boldsymbol{x})$, to fit the present dataset $\mathcal{D}$ and employ the output from the final layer of $\mathcal{M}^*(\boldsymbol{x})$ as a set of basis functions $\boldsymbol{\phi(x)}$, integrate it with a Bayesian linear regressor (BLR) to capture uncertainty; (2) sample an unobserved candidate set $\mathcal{X}$ with the explore-exploit tradeoff, select the top K seed sets with elevated acquisition values into $\mathcal{X}^*$; (3) appraise $\mathcal{X}^*$ using $\mathcal{M}(\boldsymbol{x})$ to curate a new dataset, subsequently expanding the training data $\mathcal{D}$. Upon the conclusion of the optimization cycle, the optimal $\boldsymbol{x}^*$ from the final observed dataset $\mathcal{D}$, showcasing the maximum spread number, is chosen as the definitive seed set.

\subsection{Surrogate Model}
% The message-passing graph neural networks (GNNs), as the most popular method to model geometric structure, which iteratively update node representations by aggregating information from neighboring nodes, exhibit similarities to the influence diffusion model. Nevertheless, current GNNs like GAT and GCN have limitations when it comes to handling long-range dependencies, which are crucial in the influence diffusion model. 
% In this work, we design an efficient global message-passing graph neural network as the surrogate model to learn the propagation between arbitrary nodes in the multiplex information diffusion model. 

Message-passing graph neural networks (GNNs) are widely recognized for modeling geometric structures. By iteratively updating node representations through aggregating information from adjacent nodes, they bear resemblances to the influence diffusion model. However, prevalent GNNs encounter challenges in addressing long-range dependencies, which are pivotal in the influence diffusion model. In our study, we introduce an optimized global message-passing graph neural network to learn the multiplex influence diffusion and leverage a BLR head to capture uncertainty.

\subsubsection{Global Kernelized Attention Message Passing.} 
The input seed set $\boldsymbol{x} = \{...,(u_i,v_j),...\}$ can be represented by a matrix $\mathbf{S} \in \mathbb{R}^{n \times m}$ where $\mathbf{S}_{i,j} = 1$ if $(u_i,v_j) \in \boldsymbol{x}$ and $\mathbf{S}_{i,j} = 0$ otherwise.
We first encode $\mathbf{S}$ into a low-dimensional status matrix $\mathbf{X} \in \mathbb{R}^{n \times d}$. The non-zero rows of $\mathbf{X}$ represent the status information of the seed users, while zero rows represent non-seed users. Let $\mathbf{H} \in \mathbb{R}^{n \times d}$ be the node feature matrix and $\mathbf{Z} \in \mathbb{R}^{n \times d}$ the output status matrix.
We can define a global attention message passing as follows:
\begin{equation}
\mathbf{Z} = \textrm{softmax}(\frac{\mathbf{H}\mathbf{W}_Q(\mathbf{H}\mathbf{W}_K)^\top}{\sqrt{d}})\mathbf{X}\mathbf{W}_V,
\label{eq:softmax}
\end{equation}
where $\mathbf{W}_Q$, $\mathbf{W}_K$ and $\mathbf{W}_V$ are learnable projection matrices. However, this incurs $\mathcal{O}(n^2)$ complexity, hindering scalability to large graphs. To accelerate computation, we further derive a kernel view of Equation \ref{eq:softmax}. Denoting the $i$-th row of $\mathbf{X}$, $\mathbf{H}$ and $\mathbf{Z}$ as $\mathbf{x}_i$, $\mathbf{h}_i$, $\mathbf{z}_i$ respectively, the kernelized formulation is:
\begin{equation}
\mathbf{z}_i= \sum_{j=1}^{n} \frac{\kappa(\mathbf{h}_i\mathbf{W}_Q, \mathbf{h}_j\mathbf{W}_K)}{\sum_{k=1}^{n}\kappa(\mathbf{h}_i\mathbf{W}_Q, \mathbf{h}_k\mathbf{W}_K)}(\mathbf{x}_j\mathbf{W}_V),
\label{eq:kernel}
\end{equation}
where $\kappa(\cdot,\cdot): \mathbb{R}^d \times \mathbb{R}^d \to \mathbb{R}_+$ is a positive-definite kernel measuring the pairwise similarity. 
We further randomly choose finite set of $t$ basis functions to approximate the kernel function:
\begin{equation}
\kappa(\mathbf{x},\mathbf{x}') \approx \varphi(\mathbf{x})^\top\varphi(\mathbf{x}'),
\end{equation}
\begin{equation}
\varphi(\mathbf{x}) = \frac{\textrm{exp}(\frac{-||\mathbf{x}||^2}{2})}{\sqrt{t}}[\textrm{exp}(\mathbf{w}_1^\top\mathbf{x}),\cdot\cdot\cdot,\textrm{exp}(\mathbf{w}_t^\top\mathbf{x})],
\end{equation}
where $\varphi(\cdot): \mathbb{R}^d \to \mathbb{R}^t$ is positive random feature map function and  $\mathbf{w}_k $ is independently sampled from $\mathcal{N}(0,I_d)$. It enables us to rewrite Equation \ref{eq:kernel} as follows:
\begin{equation}
\mathbf{z}_i= \sum_{j=1}^{n} \frac{\varphi(\mathbf{h}_i\mathbf{W}_Q)\varphi(\mathbf{h}_j\mathbf{W}_K)^\top}{\sum_{k=1}^{n}\varphi(\mathbf{h}_i\mathbf{W}_Q)\varphi(\mathbf{h}_k\mathbf{W}_K)^\top}(\mathbf{x}_j\mathbf{W}_V).
\end{equation}
The \textit{dot-then-exponentiate} in Equation 
\ref{eq:softmax} then converts into \textit{inner-product}, which enables two summations to be shared by each user:
\begin{equation}
\mathbf{z}_i=  \frac{\varphi(\mathbf{h}_i\mathbf{W}_Q)\sum_{j=1}^{n}\varphi(\mathbf{h}_j\mathbf{W}_K)^\top(\mathbf{x}_j\mathbf{W}_V)}{\varphi(\mathbf{h}_i\mathbf{W}_Q)\sum_{k=1}^{n}\varphi(\mathbf{h}_k\mathbf{W}_K)^\top}.
\end{equation}
Finally, we obtain the matrix form of the Global Kernelized Attention Message Passing (GKAMP) module with $\mathcal{O}(n)$ complexity:
\begin{equation}
\mathbf{Z} = \frac{\varphi(\mathbf{H}\mathbf{W}_Q)(\varphi(\mathbf{H}\mathbf{W}_K)^\top(\mathbf{X}\mathbf{W}_V))}{\textrm{diag}(\varphi(\mathbf{H}\mathbf{W}_Q)(\varphi(\mathbf{H}\mathbf{W}_K)^\top \mathbf{1}_{n\times 1}))}.
\end{equation}

\subsubsection{Basis Functions Learning.}
Firstly, we use the GKAMP module defined above to learn the complex multiplex diffusion process of the objective multiplex influence model:
\begin{equation}
\mathbf{Z}_{out} = \mathbf{X} + \textrm{GKAMP}(\mathbf{X}),
\end{equation}
where  $\mathbf{Z}_{out} \in \mathbb{R}^{n \times d}$ represents the final status matrix after multiplex information diffusion. Then we use multi-layer perception (MLP) to regress this matrix to the prediction of multiplex influence:
\begin{equation}
\hat{y} = \textrm{MLP}(\mathbf{Z}_{out}),\end{equation}
and use the Mean Absolute Error (MAE) as the loss function:
\begin{equation}
\mathcal{L} = \sum_{i=1}^{N}|y_i-\hat{y}_i|.
\end{equation}
After training, we extract $\boldsymbol{\phi}(\boldsymbol{x}) \in \mathbb{R}^D$, the output of the last hidden layer of the MLP, as the basis functions for the Bayesian linear regression. 

\subsubsection{Adaptive Basis Regression.} In this work, we construct the surrogate model by adaptively combining the basis functions $\boldsymbol{\phi}(\boldsymbol{x})$ via Bayesian linear regression (BLR).
Let $\boldsymbol{\Phi} = [\boldsymbol{\phi}(\boldsymbol{x}_1),...,\boldsymbol{\phi}(\boldsymbol{x}_N)]$ denote the design matrix arising from the training data $\mathcal{D}$, and $\mathbf{y} \in \mathbb{R}^N$ denote the stack target vector. Consider a linear regression model $y = \mathbf{w}^\top\boldsymbol{\phi}(\boldsymbol{x}) + \mathbf{b}$, where $\mathbf{w} \sim \mathcal{N}(0,\sigma_{\mathbf{w}}^2\mathbf{I})$ and $\mathbf{b} \sim \mathcal{N}(0,\sigma_{b}^2\mathbf{I})$. For a new input $\boldsymbol{x}$, the predictive mean $\mu(\boldsymbol{x}; \mathcal{D})$ and variance $\sigma^2(\boldsymbol{x}; \mathcal{D})$ of the BLR are then given
by:
\begin{equation*}
\mu(\boldsymbol{x}; \mathcal{D}) = \mathbf{m}^\top\boldsymbol{\phi}(\boldsymbol{x}),
\end{equation*}
\begin{equation*}
\sigma^2(\boldsymbol{x}; \mathcal{D}) = \boldsymbol{\phi}(\boldsymbol{x})^\top\mathbf{A}^{-1}\boldsymbol{\phi}(\boldsymbol{x}) + \sigma_{b}^2,
\end{equation*}
where
\begin{equation*}
\mathbf{m} = \sigma_{b}^{-2}\mathbf{A}^{-1}\boldsymbol{\Phi}^\top\mathbf{y},
\end{equation*}
\begin{equation*}
\mathbf{A} = \sigma_{b}^{-2}\boldsymbol{\Phi}^\top\boldsymbol{\Phi} + \mathbf{I}\sigma_{\mathbf{w}}^{-2}.
\end{equation*}

\begin{theorem}
The surrogate model constructed by combining neural network basis functions $\phi(\boldsymbol{x})$ with Bayesian linear regression is a special case of Gaussian process regression with a linear kernel.
\end{theorem}
\begin{proof}
Given a dataset $\mathcal{D}$, Bayesian linear regression over input features $\boldsymbol{\phi}(\boldsymbol{x})$ has the following posterior predictive distribution: 
\begin{equation*}
p(y|\boldsymbol{x}, \mathcal{D}) = \mathcal{N}(y|\sigma_{b}^{-2}\mathbf{A}^{-1}\boldsymbol{\Phi}^\top\mathbf{y}, \boldsymbol{\phi}(\boldsymbol{x})^\top\mathbf{A}^{-1}\boldsymbol{\phi}(\boldsymbol{x})).
\end{equation*}
It can be rewritten as follows:
\begin{equation*}
p(y|\boldsymbol{x}, \mathcal{D}) = \mathcal{N}(y|\boldsymbol{k}'^\top(\mathbf{K}+\sigma^{2}_{b}\mathbf{I})^{-1}\mathbf{y},k'' - \boldsymbol{k}'^\top(\mathbf{K}+\sigma^{2}_{b}\mathbf{I})^{-1}\boldsymbol{k}'),
\end{equation*}
where $\mathbf{K} = \boldsymbol{\Phi}\Sigma_{\mathbf{w}}\boldsymbol{\Phi}^\top$, $\boldsymbol{k}' = \boldsymbol{\Phi}\Sigma_{\mathbf{w}}\boldsymbol{\phi}(\boldsymbol{x})$ and $k'' = \boldsymbol{\phi}(\boldsymbol{x})^\top\Sigma_{\mathbf{w}}\boldsymbol{\phi}(\boldsymbol{x})$.
This is equivalent to a Gaussian process with prior mean is zero and covariance function: $\kappa(x, x') = \boldsymbol{\phi}(\boldsymbol{x})^\top\Sigma_\mathbf{w}\boldsymbol{\phi}(\boldsymbol{x}')$, where $\Sigma_\mathbf{w} = \mathbf{I}\sigma^2_{\mathbf{w}}$.
This establishes that the surrogate model using neural networks and Bayesian linear regression is a special case of Gaussian process regression with the nonlinear mapping of $\boldsymbol{\phi}(\boldsymbol{x})$.
\end{proof}

\subsection{Data Acquisition}
Data acquisition is a critical component in Bayesian Optimization. In this component, we need to trade off exploration and exploitation and quantify the promise of unobserved inputs using an acquisition function.

\subsubsection{Acquisition Function.}
In this work, we adopt the expected improvement (EI) as the acquisition function.
The EI is defined as the expectation of the improvement function $I(\boldsymbol{x}) = \textrm{max}\{0, (\mu(\boldsymbol{x};\mathcal{D})-\textrm{max}(\mathbf{y}))\}$ at candidate point $\boldsymbol{x}$. It can be formulated as:
\begin{equation*}
a_{EI}(\boldsymbol{x};\mathcal{D}) = \sigma(\boldsymbol{x}; \mathcal{D})[\gamma(\boldsymbol{x})\mathcal{C}(\gamma(\boldsymbol{x}))+\mathcal{N}(\gamma(\boldsymbol{x});0,1)],
\end{equation*}
where 
\begin{equation*}
\gamma(\boldsymbol{x}) = \frac{\mu(\boldsymbol{x};\mathcal{D})-\textrm{max}(\mathbf{y})}{\sigma(\boldsymbol{x}; \mathcal{D})}.
\end{equation*}
Here $\mathcal{C}(\cdot)$ and $\mathcal{N}(\cdot;0,1)$ denote the cumulative distribution function and probability density function of the standard normal distribution, respectively.

\subsubsection{Explore-Exploit Tradeoff.}At each optimization round, we sample candidate sets $\mathcal{X}$ and compute their acquisition values. As users and items frequently present in influential seed sets are more likely to appear in the optimal set, $\mathcal{X}$ is sampled as follows: $\alpha$ from top 5\% high influence entries, the rest uniform randomly, where $\alpha$ is the exploit ratio.
Candidates with top 1\% acquisition values are chosen as $\mathcal{X}^*$ to evaluate via the objective multiplex diffusion model, obtaining new observations $\mathcal{D}^*$ to expand the current dataset $\mathcal{D}$.

\section{Experiments}
In the following experiments, we evaluate the effectiveness of our proposed GBIM framework on four real-world networks and one synthetic network for maximizing multiplex influence across a range of seed set sizes.
\subsection{Experiment Setup}
We evaluate the expected multiplex influence defined in Equation \ref{eq:mtinf} under Multi-LT and Multi-IC, two multiplex diffusion models extended by LT and IC. We enable each multiplex diffusion model to simulate until the diffusion process stops and report the average multiplex influence over 100 simulations.
\begin{figure*}[!t]
\centering
\includegraphics[width=\textwidth]{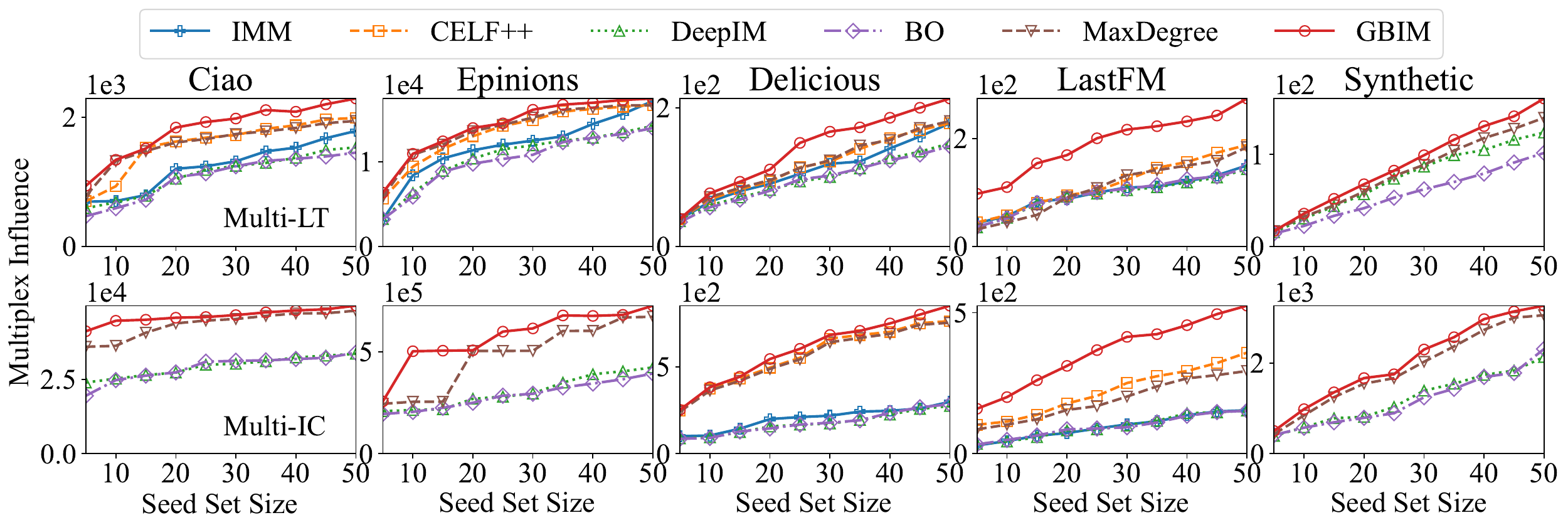}

	\caption{Performance comparison on Multi-LT (first row) and Multi-IC (second row) with the seed set size growth. The traditional approaches IMM and CELF++ exceeded time and memory limits under Multi-LT on the Synthetic dataset, and Multi-IC on Ciao, Epinions and Synthetic datasets. }
\label{fig:per}
\end{figure*}
\subsubsection{Datasets.}
We evaluate GBIM against other methods on four real-world social networking datasets: Ciao, Epinions\footnote{\url{https://www.cse.msu.edu/~tangjili/datasetcode/truststudy.htm}}, Delicious and LastFM\footnote{\url{https://grouplens.org/datasets/hetrec-2011/}}. 
These datasets containing genuine user-item interactions (i.e., ratings or frequencies) are instrumental in constructing a pragmatic preference matrix for the multiplex diffusion model.
We also use a synthetic Erdos-Renyi \cite{erdHos1960evolution} random graph with 30,000 nodes. For each real dataset, we sample some items and construct the item-item association network $\mathcal{I}$ by calculating cosine similarity over user interaction records (item pairs with similarity above 0.5 are connected). 
We use item-based collaborative filtering to generate the user preference matrix $\mathbf{P}$ from user-item interactions. 
For the synthetic dataset, we randomly construct the item network $\mathcal{I}$ and preference matrix $\mathbf{P}$. The statistics of all datasets are provided in Table \ref{tb:stat}.

\begin{table}[]

\resizebox{\linewidth}{!}{
\begin{tabular}{ccccc}
\hline
          & \#Users & \#UserEdges & \#Items & \#ItemEdges \\ \hline
Ciao      & 7317    & 170410      & 404     & 1018        \\
Epinions  & 18069   & 574064      & 411     & 1408        \\
Delicious & 1861    & 15328       & 536     & 750         \\
LastFM    & 1892    & 25434       & 501     & 906         \\
Synthetic & 30000   & 200000      & 1000    & 3000        \\ \hline
\end{tabular}}

\caption{The statistics of the datasets.}
\label{tb:stat}
\end{table}

\subsubsection{Baseline Methods.}
We compare the proposed GBIM framework with the following methods.
\begin{itemize}
\item IMM \cite{tang2015influence}: A sampling-based traditional IM algorithm that employs martingale analysis and bootstrap estimation.
\item  CELF++ \cite{goyal2011celf++}: An efficient greedy algorithm for canonical IM, which avoids unnecessary Monte Carlo simulations.
\item  DeepIM \cite{ling2023deep}: A recent learning-based IM framework, based on autoencoder and GAT, optimizing the seed set via projected gradient descent.
\item  BO: The Bayesian Optimization with Gaussian processes as the surrogate model for Multi-IM.
\item  MaxDegree: The ranking of the product of the degree of users and the degree of items.
    
\end{itemize}

\subsubsection{Implementation Details}
In multiplex influence models,  the inter-layer association strength $\beta$ is set as 0.3, and the weights of edges are set as the reciprocal of in-degree. In GBIM, the hidden dimension $d$ in GKAMP is set as 64, and we adopt a 4-layer MLP with hidden sizes 512, 1024, 1024, and 1024. The exploit rate $\alpha$ is set as 0.75. We sample 1000 instances at first and leverage the Adam optimizer \cite{kingma2014adam} with a learning rate of 0.001 for parameters learning. The experiments are implemented in a machine with the following configuration: RTX 2080 Ti GPU with 12GB VRAM, i7-9700 CPU@3.00GHz, 16GB RAM Ubuntu OS, and PyTorch 2.0.1 \cite{paszke2019pytorch}.

\begin{figure*}[htbp]
    \centering
        \subcaptionbox{Vary $n$ ($m$ = 100 and $k$ = 5)}{
        \includegraphics[width = .32\textwidth]{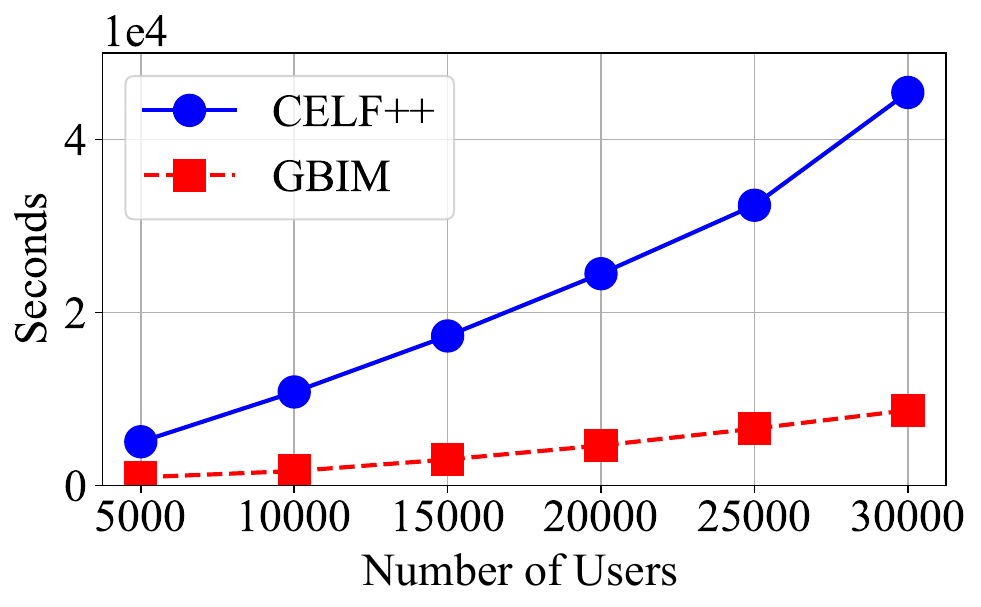}
        
    }
        \subcaptionbox{Vary $m$ ($n$ = 5000 and $k$ = 5)}{
        \centering
        \includegraphics[width = .32\textwidth]{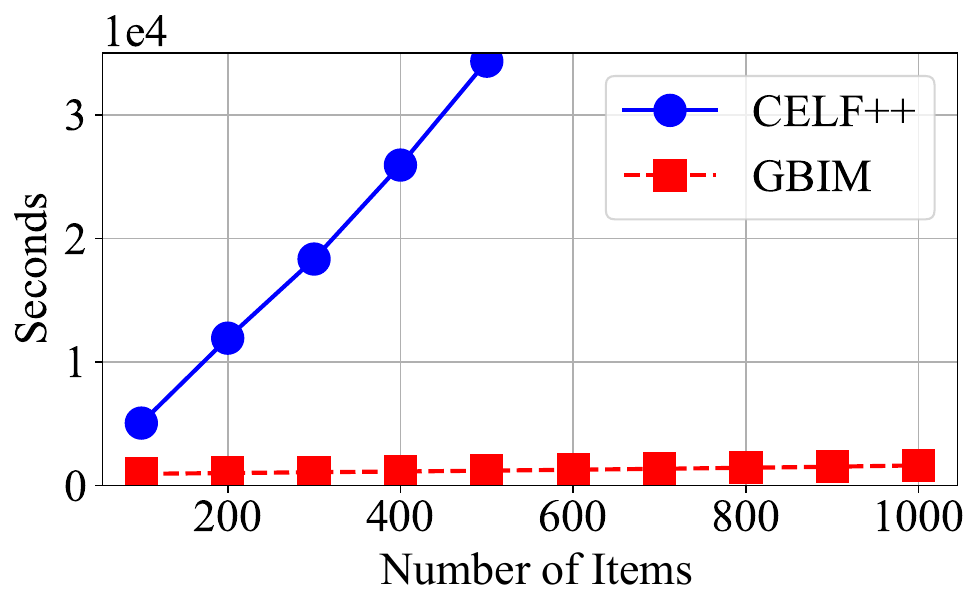}
    }
        \subcaptionbox{Vary $|\mathcal{D}|$ ($n$ = 5000, $m$ = 100 and $k$ = 5)}{
        \includegraphics[width = .32\textwidth]{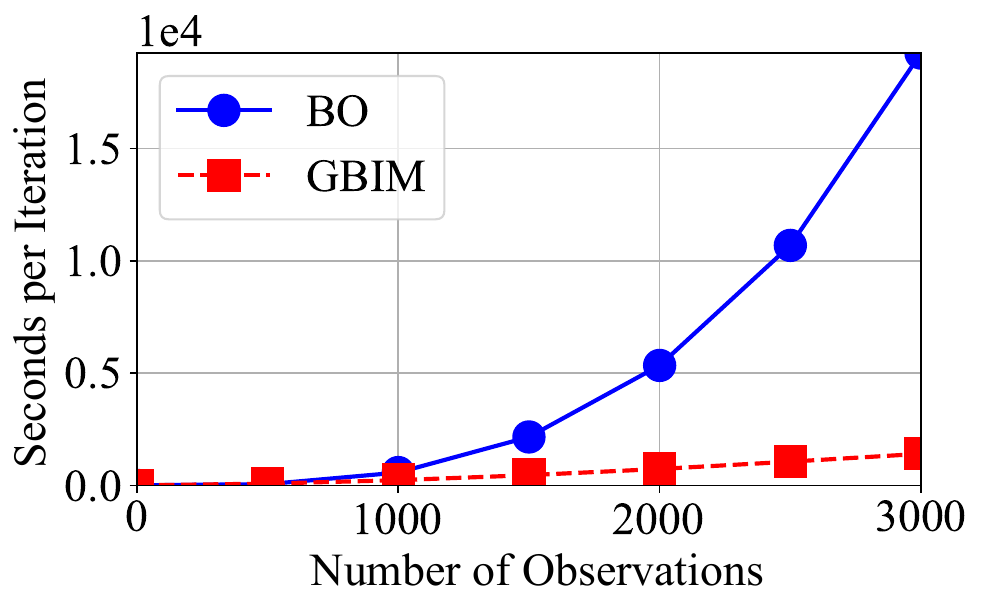}
    }
        \caption{Scalability of GBIM on the Multi-IC model of synthetic data. (a) Near-linear runtime scaling with the number of users. (b) Stable runtime as the number of items increases. (c) Linear runtime as the $|\mathcal{D}|$ increases.} 
\label{fig:scal}
\end{figure*}

\begin{figure}[t]
\centering
\includegraphics[width=0.9\linewidth]{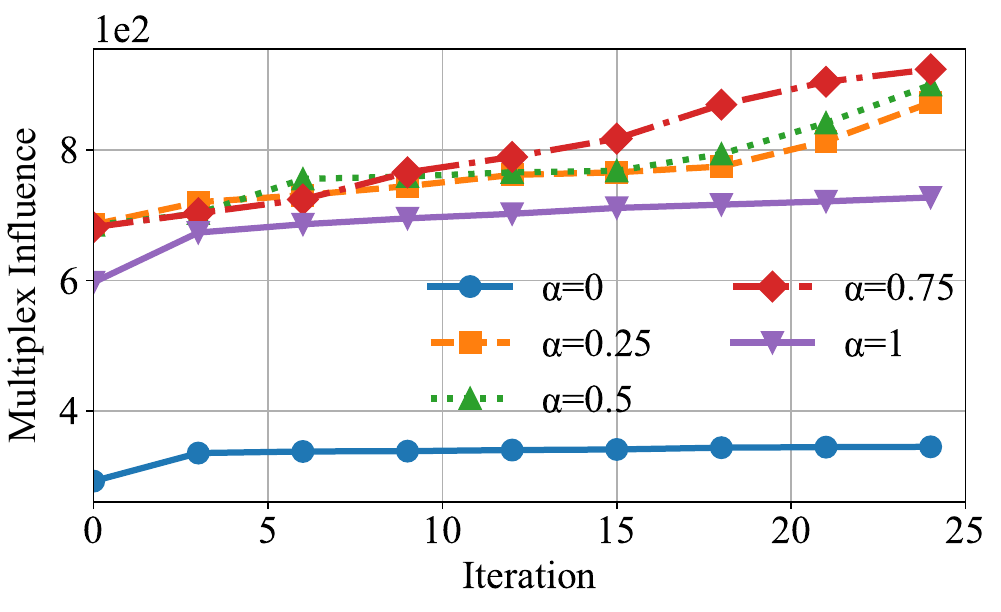}

	\caption{Performance of GBIM over iterations on the Multi-LT model of Ciao dataset, with different exploit rates $\alpha$. An appropriate value of $\alpha$ balances exploitation versus exploration, allowing GBIM to optimize efficiently.}
\label{fig:alpha}
\end{figure}

\subsection{Experimental Results}
As depicted in Figure \ref{fig:per}, GBIM consistently outperforms other methods across all experimental configurations and datasets. This substantial improvement is attributed to GBIM's Bayesian optimization strategy effectively navigating the immense search space, as well as its accurate surrogate model, which aptly captures the complex heterogeneous influence diffusion patterns and inter-layer associations.
In contrast, traditional methods like IMM and CELF++ encounter scalability bottlenecks, exceeding memory limits and runtime budgets within the Multi-LT Synthetic and Multi-IC Ciao/Epinions/Synthetic experiments. This highlights their computational constraints when applied to large-scale Multi-IM problems.
Meanwhile, the learning-based DeepIM falls short across all datasets. This can be attributed to the limitations of the GAT module they employ, which struggles to capture long-range dependencies across multiple layers. Additionally, standard Gaussian processes fail to effectively learn the intricate multiplex influence diffusion patterns, causing mediocre performance of BO methods.
Compared to these approaches, GBIM consistently achieves superior performance under seed sets of all sizes, demonstrating its robustness. For example, on the LastFM network, GBIM attained over 40\% higher multiplex influence spread than the best baseline method.

\subsection{Parameter Analysis}
In this subsection, we conduct experiments aimed at discussing the impact of the exploitation rate $\alpha$ within our data acquisition module on GBIM's optimization performance. Specifically, we assess five distinct values of $\alpha$: 0, 0.25, 0.5, 0.75, and 1. These evaluations are carried out on the Multi-LT model using the Ciao dataset, employing a seed set size of $k$=5. Other parameters remain set to their optimal configurations.
As illustrated in Figure \ref{fig:alpha}, the manipulation of exploit rates, such as $\alpha$=0.75, fosters swift initial convergence to the optimal solution. In contrast, lower values like $\alpha$=0.25 and 0.5 necessitate more iterations for the optimization process to culminate. This pattern emerges due to the higher $\alpha$ values that encourage the acquisition function to lean toward seeds akin to those previously deemed optimal, expediting the optimization. However, the scenario where $\alpha$=1 signifies complete exploitation devoid of exploration, thereby resulting in the entrapment of local optima. On the other hand, $\alpha$=0 indicates no exploitation at all, making the optimization significantly harder as the search space is not narrowed effectively.
In conclusion, an appropriate intermediate value of $\alpha$ balances exploitation and exploration, achieving efficient optimization. Empirically, $\alpha$ around 0.5 to 0.75 works well for GBIM across datasets.

\subsection{Scalability Analysis}
In this subsection, we evaluate the scalability of GBIM on the Multi-IC model using synthetic datasets. We start with a base network of 5000 users and 100 items, then progressively increase the number of users and items. As exhibited in Figure \ref{fig:scal}(a), the running time of GBIM scales nearly linearly as the number of users $n$ expands, also remaining consistently lower than CELF++. This favorable scalability is attributed to GBIM's efficient surrogate model based on global kernelized attention message passing, which embeds influence patterns into compact vector representations and performs efficient computations on the social network of size $n$, achieving significantly better scalability as $m$ grows large. In contrast, traditional IM methods perform costly simulations on the complete multiplex network of size $n \times m$, resulting in drastically higher complexity. As seen in Figure \ref{fig:scal}(b), CELF++'s runtime increases rapidly and exceeds time limits as the number of items grows, while GBIM is relatively stable. Additionally, our surrogate model based on neural networks and Bayesian regression demonstrates linear runtime growth with the training dataset size $|\mathcal{D}|$ in Figure \ref{fig:scal}(c). Whereas the standard Gaussian process BO method suffers from cubic scaling, becoming infeasible for large data.

\section{Conclusion}
In this work, we studied the intricate dynamics of concurrent multi-information propagation on directed social networks. Our core focus was solving the novel Multi-IM problem which aims to select the user-item pair into the seed set with a fixed budget to maximize multiplex influence. We first incorporated heterogeneous propagation patterns and the association mechanism into a multiplex diffusion model, where multiple information items disseminate in a multiplex network.
To address Multi-IM, we propose GBIM, a Graph Bayesian Optimization framework. We design a highly efficient global kernelized attention message-passing module to learn the complex multiplex diffusion patterns and integrate Bayesian linear regression to obtain a scalable surrogate model. Furthermore, we develop a data acquisition module with explore-exploit trade-off sampling strategy to optimize the seed set.
Extensive experiments on synthetic and real-world datasets demonstrate the scalability and effectiveness of GBIM. In the future, we will investigate how to learn complementary relations and competitive relations from data and incorporate them into a heterogeneous item-item network for a more realistic diffusion model.
% \newpage
\section*{Acknowledgements} 
This work is supported by NSFC program (No. 62272338) and NSF IIS award \#2153369.
\bibliography{aaai24}
\end{document}